\documentclass[runningheads]{llncs}
\usepackage[T1]{fontenc}

\usepackage{graphicx}
\usepackage{amsmath, amssymb}
\usepackage{enumitem}
\usepackage{hyperref}
\usepackage{color}
\usepackage{comment}
\usepackage{wrapfig}

\begin{document}
\title{Physics of Information Geometry - Part II: Small-Step Active Inference on the Probability Simplex}
%\title{Physics of Information Geometry:\\ Greedy Optimality and Lagrangian Dynamics\\on the Probability Simplex}
%
%\titlerunning{Physics of Information Geometry}
%
%\begin{comment}
\author{C.~Emre Koksal\inst{1}%\orcidID{0000-0001-8463-4446}
\and Deniz Sargun\inst{2}\thanks{This work does not relate to the author' s position at Amazon.}
%\orcidID{0000-0003-2417-7495}
%
\authorrunning{Koksal and Sargun}
\institute{The Ohio State University, Columbus OH 43210, USA,
\email{koksal.2@osu.edu} \and
Amazon.com Inc., Palo Alto CA 94301, USA, \email{denizsargun@gmail.com}}}
%
%\end{comment}
%\author{Anonymous Author(s)}
%
%\authorrunning{Koksal and Sargun}
%
%\institute{\vspace{-0.2in}}
%

\maketitle
%
%\vspace{-0.3in}
\begin{abstract}

This paper is the second in a two-part investigation of the physics of information geometry. While Part I develops a physical foundation for distributional motion on the probability simplex, the present paper studies how that framework manifests in active inference. The treatment is fully self-contained and does not require familiarity with Part I. We focus in particular on active inference through small distributional steps and the geometric structure induced by such local motion.
%In Part II, we study the geometry of active inference on the probability simplex.
Starting from an initial distribution, an agent evolves its belief state toward a final target distribution through a sequence of constrained updates. We define a relative free energy functional with respect to the preferred distribution and extend it to a relative potential energy analogous to the Helmholtz/Gibbs free-energy decomposition. The evolution is subject to a per-step kinetic constraint expressed through the Kullback-Leibler (KL) divergence between consecutive distributions, which serves as a discrete kinetic energy on the probability simplex. Using the information-geometric Pythagorean theorem on KL balls, we show that sufficiently small local moves dominate large direct jumps, and that greedy maximization of free-energy reduction is globally optimal under the kinetic constraint. This leads to a sequential variational principle in which the optimal trajectory minimizes the associated Lagrangian of the optimization problem. Similar to classical mechanics, the Lagrangian takes on the form as the difference between the kinetic and potential terms, establishing a least-action principle for distributional motion on the simplex. The resulting optimal update admits a closed form as an exponentially tilted version of the current distribution toward the preferred distribution, parametrized by an inverse-temperature-like multiplier. We further extend the framework to incorporate state-dependent geodesic costs and show that these costs are equivalent to exponential tilting of the preferred distribution itself. The proposed framework reveals a unique thermodynamic and geometric structure underlying belief evolution, connecting Friston's free-energy principle, information geometry, variational inference, and optimal transport through the physics of least-action paths on probability spaces.
\end{abstract}
\keywords{Active inference \and Information geometry \and
Free energy principle \and Kullback--Leibler divergence \and
Probability simplex \and Least Action Principle \and Pythagorean theorem}

\section{Introduction}

A fundamental question underlying active inference, learning, and adaptive decision making is how a system moves through the probability simplex from one distribution to another through a sequence of actions. In active inference~\cite{friston2010free,friston2006free}, an agent selects actions to minimize expected free energy while simultaneously steering its variational distribution toward preferred states. In diffusion-based generative models \cite{Denoising}, a system evolves from a nearly uniform or Gaussian distribution toward a structured target distribution representing an image, signal, or latent representation. Similar distributional evolution problems arise in stochastic control, Bayesian filtering, reinforcement learning, and nonequilibrium thermodynamics.  In all of these cases, the central object is not merely the final state itself, but rather the \emph{path} taken through the probability simplex.  This naturally raises the following question: \emph{what is the optimal sequence of distributions that transports a system from an initial point to a desired terminal distribution?}

Building on the thermodynamic and information-geometric foundations developed in Part I \cite{part1}, in this paper we develop a variational framework for active-inference-driven distributional motion on the probability simplex and establish its connection to the classical \textbf{principle of least action}. Starting from an initial distribution, we define a notion of \emph{relative free energy} with respect to a preferred distribution and extend it to a \emph{relative potential energy} by incorporating entropy terms in a form analogous to the Helmholtz/Gibbs free-energy decomposition~\cite{landau1980statistical}.  We then define a kinetic-energy-like quantity associated with each step of motion using the Kullback--Leibler (KL) divergence~\cite{kullback1951information,cover2006elements} between successive distributions.  Using the information-geometric Pythagorean theorem~\cite{csiszar2004information,csiszar1975divergence}, we show that taking a large leap toward the target is fundamentally suboptimal relative to taking sufficiently small incremental steps.  This leads to a \textbf{constrained sequential optimization problem} in which the kinetic energy of each step is bounded.  We prove that the resulting global problem reduces to minimizing the accumulated difference between potential and kinetic energy along the path, thereby establishing a least-action principle on the probability simplex.

Our development differs fundamentally from existing formulations in information geometry, active inference, and optimal transport. Classical information geometry~\cite{amari2016information,rao1945information,cencov1982statistical} studies geodesics induced by the Fisher metric and characterizes shortest paths on the simplex through differential geometric constructions. Schrödinger bridge problems and entropy-regularized optimal transport~\cite{benamou2000computational,villani2009optimal} construct stochastic interpolations between endpoint distributions through path-space entropy minimization.  Active inference, in contrast, minimizes expected free energy at each step in order to generate actions for preferred outcomes as well as minimum surprise~\cite{parr2022active,dacosta2020active,devries2025efe}. Closely related, KL-cost stochastic optimal control and linearly-solvable MDPs yield an exponentially tilted optimal policy of the same functional form as our update~\cite{kappen2005path,todorov2006linearly,levine2018rl}, where a state-dependent cost is absorbed by exponentially tilting a reference distribution. We have a different objective in this paper: rather than prescribing a local variational objective or a stochastic transport interpolation, we derive a sequential variational principle for \textit{distributional motion itself}.  \textbf{In our framework, free energy, entropy, and KL-based kinetic cost combine into an action functional whose minimization determines the optimal path toward a preferred distribution.}

The key mechanism enabling this construction is an information-geometric decomposition based on the Pythagorean theorem \cite{csiszar2004information,blei2017variational} for relative entropy.  Let \(q_{t-1}\) denote the current point on the simplex and let \(p_{\mathrm{pref}}\) denote the preferred distribution.  We consider the optimal next step \(q_t^\star\) constrained to lie within a small KL ball around \(q_{t-1}\). This construction is a slight variation of the KL ball introduced in Part I, where the balls were defined along a backward-in-time construction and the free-energy notion was aligned with its thermodynamic interpretation relative to Gibbs equilibrium, rather than with the Friston-type variational free energy used in the present paper. The Pythagorean relation shows that the free-energy reduction obtained by moving from \(q_{t-1}\) to \(q_t^\star\) dominates the corresponding KL step size, implying that sufficiently small local moves are more efficient than large global jumps.  Moreover, as the step size approaches zero, the reduction in free energy converges to the KL divergence between consecutive distributions up to higher-order terms.  This establishes the KL divergence as \textit{kinetic-energy} on the simplex and motivates the introduction of a least-action path formulation.

We then formulate the sequential path optimization problem explicitly.  Given an initial distribution\footnote{We use a slightly different notation from Part I, since the initial distribution in the present formulation is not required to be a Gibbs equilibrium distribution.} \(\pi_0\) and a preferred distribution \(p_{\mathrm{pref}}\), the objective is to maximize the total reduction in relative free energy subject to a per-step kinetic-energy constraint.  We show that the problem admits a greedy optimal structure in which each infinitesimal step locally maximizes free-energy reduction.  This converts the original constrained optimization into the minimization of an accumulated action functional over the entire trajectory.  The resulting optimal update admits a closed-form characterization: at each step, the new distribution is obtained by exponentially tilting the current distribution toward the preferred distribution.  Thus, the least-action path evolves through a sequence of progressively tilted distributions converging toward the destination.

Finally, we extend the framework to incorporate geodesic or state-dependent costs on the simplex.  These costs represent %energetic penalties, environmental constraints, or 
preferences associated with occupying specific regions of the simplex.  We show that introducing such costs is equivalent to exponentially tilting the preferred distribution itself by the corresponding cost function.  Consequently, the generalized problem reduces to the original least-action formulation under a transformed preferred distribution.  This reveals a deep structural equivalence between external geometric costs and exponential tilting in information geometry.

\section{Model}

We consider sequential motion on the probability simplex
\[
\Delta(\mathcal X)
=
\left\{
q \in \mathbb R_{\ge 0}^{|\mathcal X|}
:
\sum_{x\in\mathcal X} q(x)=1
\right\}, \]
where an agent evolves its belief state from an initial equilibrium
distribution toward a preferred target distribution through a sequence
of actions.  Throughout the paper, \(\pi_0\) denotes the initial arbitrary distribution and \(p_{\mathrm{pref}}\) denotes the preferred
distribution that the agent seeks to reach.  The sequence
\(\{q_t\}_{t\ge 0}\) describes the trajectory of the agent on the
probability simplex.

\begin{comment}
Our objective is not merely to characterize the final distribution,
but rather to understand the \emph{physics of the path itself}.
To this end, we introduce analogues of free energy, potential energy,
and kinetic energy directly on the simplex.  These quantities reveal a
variational structure underlying belief evolution and eventually lead
to a least-action principle for optimal motion between probability
distributions.
\end{comment}

%\subsection{Relative Free Energy, Potential Energy, and Kinetic Energy}

Next, we provide the fundamental definitions on the probability simplex. We begin with the notion of relative free energy.

\vspace{-0.05in}

\begin{definition}[Relative free energy]
The relative free energy of the current distribution \(q_t\)
with respect to the equilibrium distribution \(\pi_0\)
and preferred distribution \(p_{\mathrm{pref}}\) is\footnote{From Part I to the present paper, we modify the notation for relative free energy to emphasize the fundamental distinction between the thermodynamic and active-inference formulations, in particular the reversal in the ordering of the distributions in the KL divergence.}
\begin{equation}
F_{\mathrm{rel}}^{\mathrm{pref}}(q_t\|\pi_0)
=
D_{\mathrm{KL}}(\pi_0\|p_{\mathrm{pref}})
-
D_{\mathrm{KL}}(q_t\|p_{\mathrm{pref}}).
\label{eq:relative_free_energy}
\end{equation}
\end{definition}
The interpretation of \eqref{eq:relative_free_energy} is important.
The quantity \(D_{\mathrm{KL}}(q_t\|p_{\mathrm{pref}})\) is itself
Friston's variational free energy relative to the preferred state.
Instead of measuring free energy absolutely, however, we measure it
\emph{relative to the initial state} \(\pi_0\).  The term
\(D_{\mathrm{KL}}(\pi_0\|p_{\mathrm{pref}})\) therefore represents the
initial free-energy gap between the initial and the preferred state,
while \(F_{\mathrm{rel}}^{\mathrm{pref}}(q_t\|\pi_0)\) quantifies how
much of that gap has been reduced by moving from \(\pi_0\) to \(q_t\).
The relative free energy thus measures progress toward the preferred
distribution in information-geometric terms.

%Next, we define the corresponding relative potential energy.
\vspace{-0.05in}

\begin{definition}[Relative potential energy]
The relative potential energy of a distribution \(q_t\) with respect to
the equilibrium distribution \(\pi_0\) and preferred distribution
\(p_{\mathrm{pref}}\) is
\begin{equation}
P_{\mathrm{rel}}^{\mathrm{pref}}(q_t\|\pi_0)
\triangleq
H(\pi_0\|p_{\mathrm{pref}})
-
H(q_t\|p_{\mathrm{pref}}),
\label{eq:relative_potential}
\end{equation}
where
\(
H(q\|p) = -\sum_{x\in\mathcal X} q(x)\log p(x)
\)
denotes the cross-entropy.% between \(q\) and \(p\).
\end{definition}

The definition above is motivated by the classical
Helmholtz/Gibbs free-energy decomposition~\cite{jaynes1957information}. Indeed, in thermodynamics,
free energy combines an energetic contribution and an entropic
contribution:
\( D_{\mathrm{KL}}(q\|p) = H(q\|p)-H(q). \)  Here,
\begin{comment}
the cross-entropy \( H(q_t\|p_{\mathrm{pref}}) \)
plays the role of a generalized potential energy relative to the
preferred distribution.  The quantity
\(
H(\pi_0\|p_{\mathrm{pref}})
\)
therefore represents the initial total potential associated with the
equilibrium distribution, while
\end{comment}
\( P_{\mathrm{rel}}^{\mathrm{pref}}(q_t\|\pi_0) \)
measures the reduction in this potential achieved by moving from
\(\pi_0\) to \(q_t\).  In this sense, the relative potential energy
quantifies how much of the initial energetic mismatch to the preferred
distribution has been eliminated.
\begin{comment}
The definition is also closely connected to the decomposition of
relative free energy.  Recalling that
\(
D_{\mathrm{KL}}(q\|p)
=
H(q\|p)-H(q),
\)
we may write
\begin{equation}
P_{\mathrm{rel}}^{\mathrm{pref}}(q_t\|\pi_0) = 
%&= H(\pi_0\|p_{\mathrm{pref}}) - H(q_t\|p_{\mathrm{pref}}) 
%\nonumber\\
%&= D_{\mathrm{KL}}(\pi_0\|p_{\mathrm{pref}}) + H(\pi_0) - \Big( D_{\mathrm{KL}(q_t\|p_{\mathrm{pref}}) + H(q_t) \Big)
%\nonumber\\
%&= 
F_{\mathrm{rel}}^{\mathrm{pref}}(q_t\|\pi_0) + H(\pi_0)-H(q_t).
\label{eq:potential_free_energy_relation}
\end{equation}
Equation \eqref{eq:potential_free_energy_relation} reveals the precise
relationship between relative free energy and relative potential
energy.  The two differ only by the entropy change between the current
distribution and the initial equilibrium distribution.  Consequently,
\end{comment}
The relative free energy may be interpreted as the portion of the
potential energy that remains after accounting for the entropy of the
current state.

\begin{comment}

An important special case occurs when the equilibrium distribution
\(\pi_0\) is uniform.  In that case,
\[
H(\pi_0)-H(q_t)
=
D_{\mathrm{KL}}(q_t\|\pi_0),
\]
and therefore
\begin{equation}
P_{\mathrm{rel}}^{\mathrm{pref}}(q_t\|\pi_0)
=
F_{\mathrm{rel}}^{\mathrm{pref}}(q_t\|\pi_0)
+
D_{\mathrm{KL}}(q_t\|\pi_0).
\label{eq:uniform_relation}
\end{equation}

Thus, when the system starts from a maximally uncertain equilibrium,
the additional contribution between potential energy and free energy
is exactly the informational distance from the current distribution to
the equilibrium point itself.

We next define the kinetic energy associated with motion on the
probability simplex.
\end{comment}

\vspace{-0.05in}

\begin{definition}[Kinetic energy]
The kinetic energy of a transition from \(q_{t-1}\) to \(q_t\) is
defined as
\begin{equation}
K(q_t,q_{t-1})
=
D_{\mathrm{KL}}(q_t\|q_{t-1})
=
\sum_{x\in\mathcal X}
q_t(x)\log\frac{q_t(x)}{q_{t-1}(x)}.
\label{eq:kinetic_energy}
\end{equation}
\end{definition}
The KL divergence between consecutive distributions measures the amount
of informational motion required to move from one point on the simplex
to another.  For infinitesimal changes, the KL divergence reduces to
the Fisher--Rao metric~\cite{rao1945information,cencov1982statistical} and therefore acts as the natural local notion
of squared speed on the simplex.  In this sense,
\eqref{eq:kinetic_energy} plays the role of a discrete kinetic energy:
large abrupt changes in belief incur high kinetic cost, while smooth
evolution corresponds to low kinetic expenditure.
The asymmetry of the KL divergence is also important here.
The quantity \(D_{\mathrm{KL}}(q_t\|q_{t-1})\) measures the cost of
\emph{selecting} the new distribution \(q_t\) relative to the previous
belief state \(q_{t-1}\), making it the natural directional notion of motion for sequential decision making. In nonequilibrium statistical mechanics, the same quantity has been identified with dissipated work, providing an independent physical justification for treating KL divergence as an energetic cost of distributional motion \cite{kawai2007dissipation}.

All of the information-geometric quantities defined above are dimensionless.  Multiplication by \(k_B \tau\) assigns them physical energy units. A direct 
thermodynamic interpretation requires additional structure: if the final distribution \(p_{\mathrm{pref}}\) is identified with an equilibrium Gibbs distribution at temperature \(\tau\), then
\(
k_B \tau\, D_{\mathrm{KL}}(q\|p_{\mathrm{pref}})
\)
corresponds to the nonequilibrium free-energy excess of \(q\) relative to that Gibbs state. Under this interpretation, the relative free- and potential-energy  quantities defined above acquire their corresponding physical energy meanings. Likewise,
\(
k_B \tau\, D_{\mathrm{KL}}(q_t\|q_{t-1})
\)
has units of energy, although its interpretation as kinetic energy remains information-geometric rather than conventional mechanical kinetic energy.
\section{Main Insight: Information-Geometric Pythagorean Theorem and Incremental Motion}

In this section, we establish the fundamental geometric principle
underlying the proposed framework.  The key observation is that,
on the probability simplex, moving toward a preferred distribution
through sufficiently small local steps is fundamentally more efficient
than taking large direct jumps.  This result follows directly from the
information-geometric Pythagorean theorem for relative entropy and
provides the basis for the least-action formulation developed in the
subsequent sections.

\begin{comment}
\begin{wrapfigure}{r}{2.8in}
\vspace{-0.4in}
\centerline{\includegraphics[width=2.8in]{pythagorean_simplex.png}}
\vspace{-0.16in}
\caption{\small Illustration of the information-geometric Pythagorean theorem on the probability simplex.  The admissible next-step distributions lie inside the KL ball
\(
\mathcal B_\delta(q_{t-1})
\).
The optimal incremental step \(q_t^\star\) is the projection of the preferred distribution \(p_{\mathrm{pref}}\) onto this ball.  The resulting decomposition shows that the reduction in free energy is at least as large as the informational step size.}
\label{fig:pythagorean}
\vspace{-0.3in}
\end{wrapfigure}
\end{comment}
\begin{figure}[t]
\centering
\includegraphics[width=0.6\textwidth]{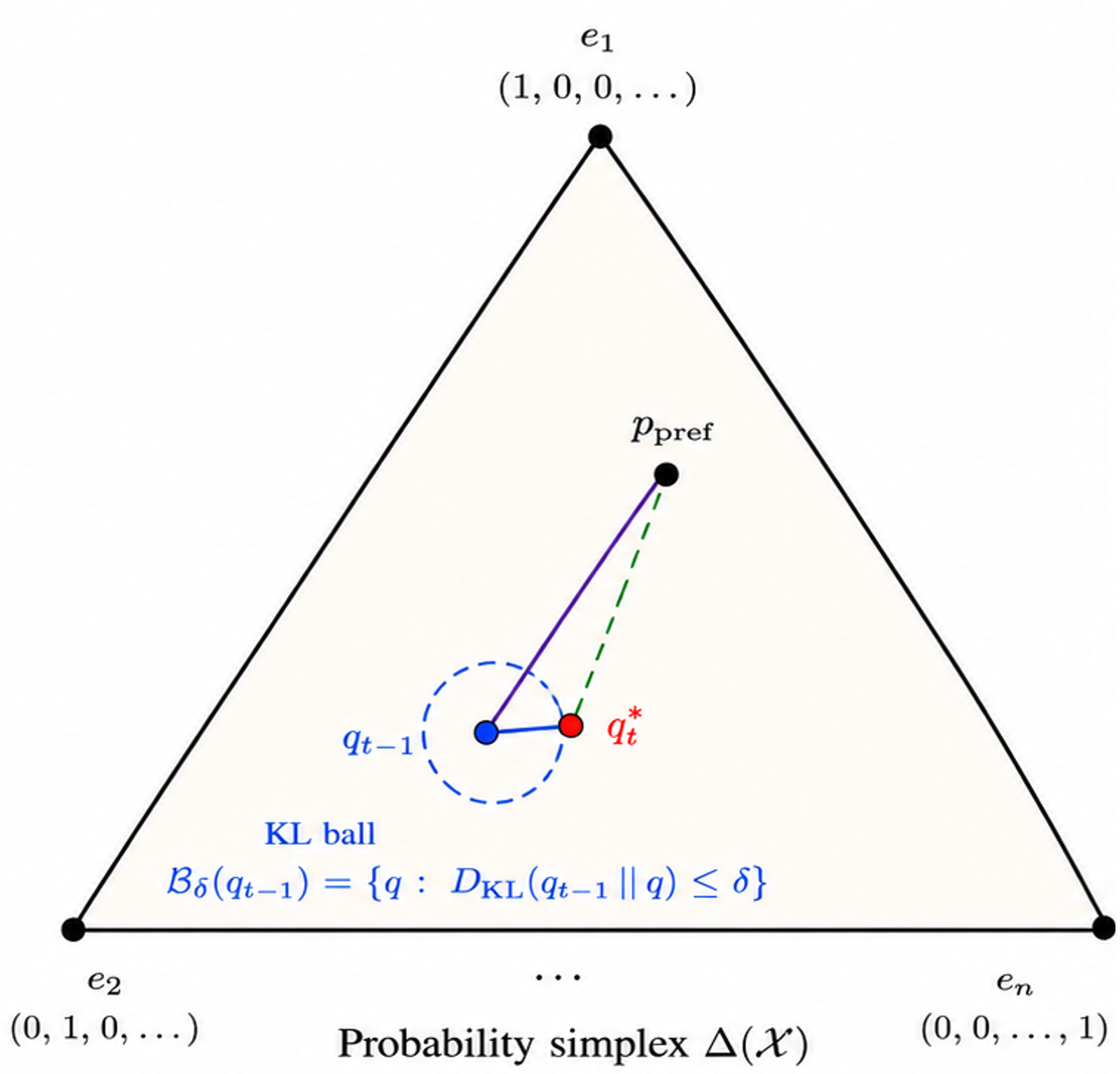}
\vspace{-0.1in}
\caption{
Illustration of the information-geometric Pythagorean theorem on the
probability simplex.  The admissible next-step distributions lie
inside the KL ball
\(
\mathcal B_\delta(q_{t-1})
\).
The optimal incremental step \(q_t^\star\) is the projection of the
preferred distribution \(p_{\mathrm{pref}}\) onto this ball.  The
resulting decomposition shows that the reduction in free energy is at
least as large as the informational step size.
}
\label{fig:pythagorean}
\vspace{-0.1in}
\end{figure}

\subsection{KL Balls and Optimal Incremental Motion}

Fix the current distribution \(q_{t-1}\) and preferred distribution
\(p_{\mathrm{pref}}\).  Consider the KL ball
\begin{equation}
\mathcal B_\delta(q_{t-1})
=
\left\{
q \in \Delta(\mathcal X)
:
D_{\mathrm{KL}}(q_{t-1}\|q)\le \delta
\right\},
\label{eq:kl_ball}
\end{equation}
which represents the set of all admissible next-step distributions
reachable under the kinetic-energy budget \(\delta\).

We define the optimal next step as
\begin{equation}
q_t^\star
=
\underset{q\in\mathcal B_\delta(q_{t-1})}{\text{argmin}} \ ~
D_{\mathrm{KL}}(q\|p_{\mathrm{pref}}).
\label{eq:optimal_projection}
\end{equation}
Geometrically, \(q_t^\star\) is the information projection of
\(p_{\mathrm{pref}}\) onto the KL ball centered at \(q_{t-1}\).  This
construction is illustrated conceptually in Fig.~\ref{fig:pythagorean}.
The figure depicts the simplex, the preferred distribution
\(p_{\mathrm{pref}}\), the current distribution \(q_{t-1}\), and the
KL ball of admissible next states.  The point \(q_t^\star\) lies on
the boundary of the ball in the direction of the preferred
distribution and corresponds to the locally optimal incremental move. The following theorem is the central geometric observation.

\begin{comment}
%\vspace{-0.2in}
\begin{figure}[t]
\centering
\includegraphics[width=0.6\textwidth]{pythagorean_simplex.png}
\vspace{-0.2in}
\caption{
Illustration of the information-geometric Pythagorean theorem on the
probability simplex.  The admissible next-step distributions lie
inside the KL ball
\(
\mathcal B_\delta(q_{t-1})
\).
The optimal incremental step \(q_t^\star\) is the projection of the
preferred distribution \(p_{\mathrm{pref}}\) onto this ball.  The
resulting decomposition shows that the reduction in free energy is at
least as large as the informational step size.
}
\label{fig:pythagorean}
\vspace{-0.2in}
\end{figure}
\end{comment}

\begin{theorem}[Information-geometric Pythagorean theorem]
\label{thm:pyth}
Let \(q_t^\star\) be defined by \eqref{eq:optimal_projection}.  Then
\begin{equation}
D_{\mathrm{KL}}(q_{t-1}\|p_{\mathrm{pref}})
-
D_{\mathrm{KL}}(q_t^\star\|p_{\mathrm{pref}})
\ge
D_{\mathrm{KL}}(q_{t-1}\|q_t^\star).
\label{eq:pythagorean}
\end{equation}
\end{theorem}

\begin{proof}
The result follows directly from the Pythagorean theorem for Bregman
divergences applied to the KL divergence.  Since the KL ball
\(
\mathcal B_\delta(q_{t-1})
\)
is convex and \(q_t^\star\) is the information projection of
\(p_{\mathrm{pref}}\) onto this set, we have
\[
D_{\mathrm{KL}}(q_{t-1}\|p_{\mathrm{pref}})
\ge
D_{\mathrm{KL}}(q_{t-1}\|q_t^\star)
+
D_{\mathrm{KL}}(q_t^\star\|p_{\mathrm{pref}}),
\]
which immediately yields \eqref{eq:pythagorean}.
\end{proof}

Equation \eqref{eq:pythagorean} has a direct physical interpretation.
The quantity
\(
D_{\mathrm{KL}}(q_{t-1}\|p_{\mathrm{pref}})
-
D_{\mathrm{KL}}(q_t^\star\|p_{\mathrm{pref}})
\)
is precisely the gain in relative free energy achieved by moving from
\(q_{t-1}\) to \(q_t^\star\).  The theorem therefore states that the
gain in free energy is lower bounded by the informational distance
traveled during the step.

This observation immediately implies that large direct jumps toward
the preferred distribution are inefficient.  Instead, the optimal
strategy is to repeatedly move to the best local point within a small
KL neighborhood.  \textit{The smaller the admissible kinetic budget
\(\delta\), the more efficiently free energy is converted into useful
motion toward the preferred state.}

\subsection{Small-Step Limit and Symmetry of KL Divergence}

The theorem above lower bounds the free-energy gain by
\(
D_{\mathrm{KL}}(q_{t-1}\|q_t^\star).
\)
However, our kinetic energy was defined as
\(
K_t
=
D_{\mathrm{KL}}(q_t^\star\|q_{t-1}),
\)
which is the reverse KL divergence.  Since the KL divergence is
asymmetric, the gain is not immediately identical to the kinetic
energy.  Nevertheless, in the infinitesimal regime the two quantities
become asymptotically equivalent. Next, we formalize this statement.

\vspace{-0.05in}

\begin{lemma}[Local symmetry of KL divergence]
\label{lemma:KL_symmetry}
Let \( q_t^\star(x) = q_{t-1}(x)+\varepsilon_x\), with $\sum_x \varepsilon_x =0$. Then,
%with \(\|\varepsilon\|\) sufficiently small.
\begin{equation}
D_{\mathrm{KL}}(q_t^\star\|q_{t-1})
=
D_{\mathrm{KL}}(q_{t-1}\|q_t^\star)
+
O(\|\varepsilon\|^3).
\label{eq:kl_symmetry}
\end{equation}
Further, let us impose an upper bound on the kinetic energy: $D(g^*_t || q_{t-1})\leq \delta$ and define \( q_{\min}\triangleq \min_{x\in\mathcal X}\ q_{t-1}(x)>0. \) If $q_{\min}\geq \sqrt{8\delta}$, then 
\begin{equation}
D_{\mathrm{KL}}(q_t^\star\|q_{t-1})
=
D_{\mathrm{KL}}(q_{t-1}\|q_t^\star)
+
O(\delta^{3/2}).
\label{eq:kl_symmetry_delta}
\end{equation}
\end{lemma}

\begin{proof}
See Appendix~\ref{sec:proof_lemma_KL_sym}.
\end{proof}

The significance of Lemma~\ref{lemma:KL_symmetry} is substantial. The kinetic energy and
the free-energy gain become identical to second order in the
infinitesimal limit.  Consequently, as \(\delta\to 0\),
\(
D_{\mathrm{KL}}(q_t^\star\|q_{t-1})
\approx
D_{\mathrm{KL}}(q_{t-1}\|q_t^\star),
\)
and therefore the free-energy gain per step becomes asymptotically
lower bounded by the kinetic budget itself:
%\begin{equation}
\( D_{\mathrm{KL}}(q_{t-1}\|p_{\mathrm{pref}})
-
D_{\mathrm{KL}}(q_t^\star\|p_{\mathrm{pref}})
\gtrsim
\delta. \)
%\label{eq:gain_delta}
%\end{equation}

This result establishes an important physical interpretation of the
proposed framework.  In the infinitesimal regime, the reduction in
free energy induced by an optimal step becomes equal to the kinetic
energy expended during that step.  The system therefore evolves along
a path in which informational motion and free-energy dissipation are
locally balanced. This balance forms the basis for the least-action principle derived
in the next section.  Rather than selecting arbitrary trajectories on
the simplex, the optimal evolution is obtained by accumulating
infinitesimal steps that maximally reduce free energy relative to the
kinetic energy required to realize them.
\section{Problem Statement}
\label{sec:problem}

Starting from an
initial equilibrium distribution \(\pi_0\), the agent seeks to move
toward the preferred distribution \(p_{\mathrm{pref}}\) through a
sequence of constrained updates
\(
\pi_0 \rightarrow q_1 \rightarrow q_2 \rightarrow \cdots \rightarrow q_N.
\)
The information-geometric Pythagorean theorem shows that small local moves are fundamentally more efficient than large direct jumps, while the small-step symmetry of KL divergence establishes that, in the infinitesimal regime, the gain in free energy becomes locally balanced by the kinetic energy expended during the step.

Accordingly, we impose the per-step constraint $K(q_t,q_{t-1}) \leq \delta$ for $t=1,\dots,N$,
\begin{comment}
\begin{equation}
K(q_t,q_{t-1}) =
D_{\mathrm{KL}}(q_t\|q_{t-1})
\le \delta,
\qquad t=1,\dots,N,
\label{eq:kinetic_constraint}
\end{equation}
\end{comment}
where \(\delta>0\) determines the maximum allowable motion in a single step. Geometrically, this constraint defines a KL ball centered at the previous distribution \(q_{t-1}\). 

Now, we define a stopping criterion\footnote{We use a different notation from Part I to reflect the direction of the construction: whereas Part I builds the trajectory backward from the preferred distribution toward the initial reference state, the present paper evolves the distribution forward from the initial state toward the preferred distribution.} to reach the preferred distribution. 
\begin{equation}
N_\delta
=
\inf
\left\{
t\ge 1:
D_{\mathrm{KL}}(p_{\mathrm{pref}}\|q_{t-1})\le \delta
\right\}.
\label{eq:stopping_time}
\end{equation}
The interpretation of \eqref{eq:stopping_time} is simple.  The
sequence of greedy least-action steps continues as long as the
preferred distribution lies outside the admissible KL ball centered at
the current distribution.  Once
\(
p_{\mathrm{pref}}\in
\mathcal B_\delta(q_{N_\delta-1}),
\)
the destination becomes reachable in a single step, and we terminate
the trajectory by setting
%\begin{equation}
\( q_{N_\delta}=p_{\mathrm{pref}} \).
%\label{eq:terminal_condition}
%\end{equation}
This stopping rule preserves the local least-action structure while guaranteeing exact arrival at the preferred distribution whenever it becomes admissible under the kinetic-energy budget. Next, we provide the main problem statement.

\vspace{-0.05in}

\begin{definition}[Main Problem: Greedy least-action step]
Given the current distribution \(q_{t-1}\), for all $t\leq N_\delta$, the next point \(q_t\) is chosen as the solution of
\begin{align}
\label{eq:greedy_problem-obj}
q_t^\star
&=
\underset{q\in\Delta(\mathcal X)}{\text{argmax}} \ ~
P_{\mathrm{rel}}^{\mathrm{pref}}(q\|q_{t-1}) \\
&\text{subject to} ~\
D_{\mathrm{KL}}(q\|q_{t-1})\le \delta.
\label{eq:greedy_problem-constraint}
\end{align}
\end{definition}
The information-geometric Pythagorean theorem shows that
sufficiently small incremental moves are more efficient than large
global jumps. The small-step limit reveals that free-energy gain and kinetic
energy become locally balanced. The optimal path is therefore constructed as a sequence of
locally optimal least-action steps, each maximizing potential gain
under a local kinetic constraint. The process terminates once the preferred distribution becomes
reachable in a single admissible step. Hence, \textbf{the action} defined as the trajectory from \(\pi_0\) to
\(p_{\mathrm{pref}}\) emerges as a
repeated application of a local information-theoretic least-action
principle, leading to the least ``cumulative'' action at the end of the process.
\section{Solution: Lagrangian and the Principle of Least Action}
\label{sec:greedy}

\subsection{Lagrangian and Least Action}

Let $\lambda_t\geq 0$ denote the Lagrange multiplier associated with
the kinetic-energy constraint. The Lagrangian for this maximization
problem is
\begin{equation}\label{eq:lagrangian}
\mathcal{L}_t(q,\lambda_t)
=
P_{\mathrm{rel}}^{\mathrm{pref}}(q\|q_{t-1})
-
\lambda_t
\left[
D_{\mathrm{KL}}(q\|q_{t-1})-\delta
\right].
\end{equation}
The additive term $\lambda_t\delta$ does not depend on $q$.
Consequently, for a fixed multiplier $\lambda_t$, maximizing
\eqref{eq:lagrangian} over $q$ is equivalent to minimizing
\begin{equation}\label{eq:instantaneous_action}
\mathcal{A}_t(q;q_{t-1})
\triangleq
\lambda_t
\underbrace{D_{\mathrm{KL}}(q\|q_{t-1})}_{\text{kinetic energy}} ~ ~
-
\underbrace{
P_{\mathrm{rel}}^{\mathrm{pref}}(q\|q_{t-1})
}_{\text{relative potential-energy}} .
\end{equation}
We refer to $\mathcal{A}_t$ as the instantaneous
information-geometric action \cite{wibisono2016variational,chirco2020lagrangian}. It has the same kinetic-minus-potential
algebraic structure as the classical Lagrangian: the first term penalizes
distributional motion, while the second rewards progress in the potential
landscape induced by the preferred distribution.

Thus, each step balances two competing effects. The kinetic term resists
abrupt changes from the current belief state, whereas the potential
term drives the distribution toward states favored by
$p_{\mathrm{pref}}$. 
%The action in~\eqref{eq:instantaneous_action} is dimensionless under our normalization. Multiplication by $k_B T$ assigns the two terms physical energy units; if a physical duration $\Delta t$ is associated with each update, then $k_B T\,\Delta t\,\mathcal{A}_t$ has the units of mechanical action. 
Note that, the term ``least action'' refers here to the variational structure of distributional motion on the probability simplex, rather than to an assertion of microscopic Newtonian dynamics.

%The resulting path construction has a natural physical interpretation. The original global problem identifies the preferred terminal distribution but does not determine the local mechanics of motion. The information-geometric Pythagorean theorem shows that sufficiently small incremental moves are more efficient than large global jumps. The small-step limit reveals that free-energy gain and kinetic energy become locally balanced. The optimal path is therefore constructed as a sequence of locally optimal least-action steps, each maximizing potential gain under a local kinetic constraint. The process terminates once the preferred distribution becomes reachable in a single admissible step.
\subsection{Least-Action Path: Tilted Distribution Dynamics}
\label{sec:solution}

We now solve the local least-action problem derived in the previous
section. Using
\begin{equation}
P_{\mathrm{rel}}^{\mathrm{pref}}(q\|q_{t-1})
=
H(q_{t-1}\|p_{\mathrm{pref}})
-
H(q\|p_{\mathrm{pref}}),
\end{equation}
where the first term is constant with respect to $q$, the action can
equivalently be written, up to an additive constant independent of $q$,
as
\begin{equation}\label{eq:action_cross_entropy}
\mathcal{A}_t(q;q_{t-1})
\equiv
H(q\|p_{\mathrm{pref}})
+
\lambda_t D_{\mathrm{KL}}(q\|q_{t-1}).
\end{equation}
Hence, the greedy least-action step is characterized by
\begin{equation}\label{eq:least_action_step}
q_t^\star
= \underset{q\in\Delta(\mathcal X)}{\text{argmin}} ~
\left\{
H(q\|p_{\mathrm{pref}})
+
\lambda_t D_{\mathrm{KL}}(q\|q_{t-1})
\right\},
\end{equation}
which means that the optimal next distribution is found as the outcome of the \textbf{least action}. For the non-terminal greedy steps, when
$p_{\mathrm{pref}}$ lies outside the admissible KL ball, the kinetic
constraint is active and $\lambda_t>0$ is selected so that
\[ D_{\mathrm{KL}}(q_t^\star\|q_{t-1})=\delta. \]
Once the preferred distribution becomes reachable within the kinetic
budget, the stopping rule is invoked and the trajectory terminates at
$p_{\mathrm{pref}}$. Problem~\eqref{eq:least_action_step} is structurally identical to a per-step \emph{fully probabilistic design} problem~\cite{KARNY2006259,KARNY19961719} and to the KL-regularized decision problem of information-theoretic \emph{bounded rationality}~\cite{ortega2013thermodynamics}, both of which yield exponentially tilted decision rules of the same form as the update in Theorem~\ref{thm:tilted}. The distinctive contribution of the present work is embedding this local update into a sequential least-action geometry, where the Pythagorean inequality (Theorem~\ref{thm:pyth}) guarantees the greedy step is globally optimal under the kinetic constraint.

The following theorem gives the closed-form solution.

\vspace{-0.05in}

\begin{theorem}[Tilted distribution update]
\label{thm:tilted}
Let \(q_{t-1} \) %\in\Delta(\mathcal X)\)
and \(p_{\mathrm{pref}}\) be the full-support current and preferred distribution, respectively. The solution of \eqref{eq:least_action_step} is:% identical to the previous distribution, tilted by the preferred distribution:
\begin{equation}
q_t^\star(x)
=
\frac{
q_{t-1}(x)\,p_{\mathrm{pref}}(x)^{1/\lambda}
}{
\sum_{y\in\mathcal X}
q_{t-1}(y)\,p_{\mathrm{pref}}(y)^{1/\lambda}
}.
\label{eq:tilted_distribution}
\end{equation}
\end{theorem}

\begin{proof}
See Appendix~\ref{sec:proof_thm_tilted}.
\end{proof}

%\subsection{Interpretation}
\vspace{-0.1in}

\noindent \textbf{Interpretation:} Solution given in \eqref{eq:tilted_distribution} shows that the optimal next distribution is obtained by a multiplicative deformation of the current distribution:
\( q_t^\star(x)
\propto
q_{t-1}(x)p_{\mathrm{pref}}(x)^{1/\lambda}.
\)
Thus, the preferred distribution acts as an exponential tilt.  States
with large \(p_{\mathrm{pref}}(x)\) are amplified, while states with
small \(p_{\mathrm{pref}}(x)\) are suppressed.  The multiplier
\(\lambda\) controls the strength of the tilt.
When \(\lambda\) is large,
\(
p_{\mathrm{pref}}(x)^{1/\lambda}\approx 1,
\)
so the update is conservative and \(q_t^\star\) remains close to
\(q_{t-1}\).  When \(\lambda\) is small, the tilt becomes stronger and
the distribution moves more aggressively toward the preferred state.
In this sense, \(\lambda\) plays the role of an inverse-temperature-like
parameter governing the tradeoff between kinetic resistance and
potential attraction.

\begin{figure}[h]
\centering
\includegraphics[width=0.48\columnwidth]{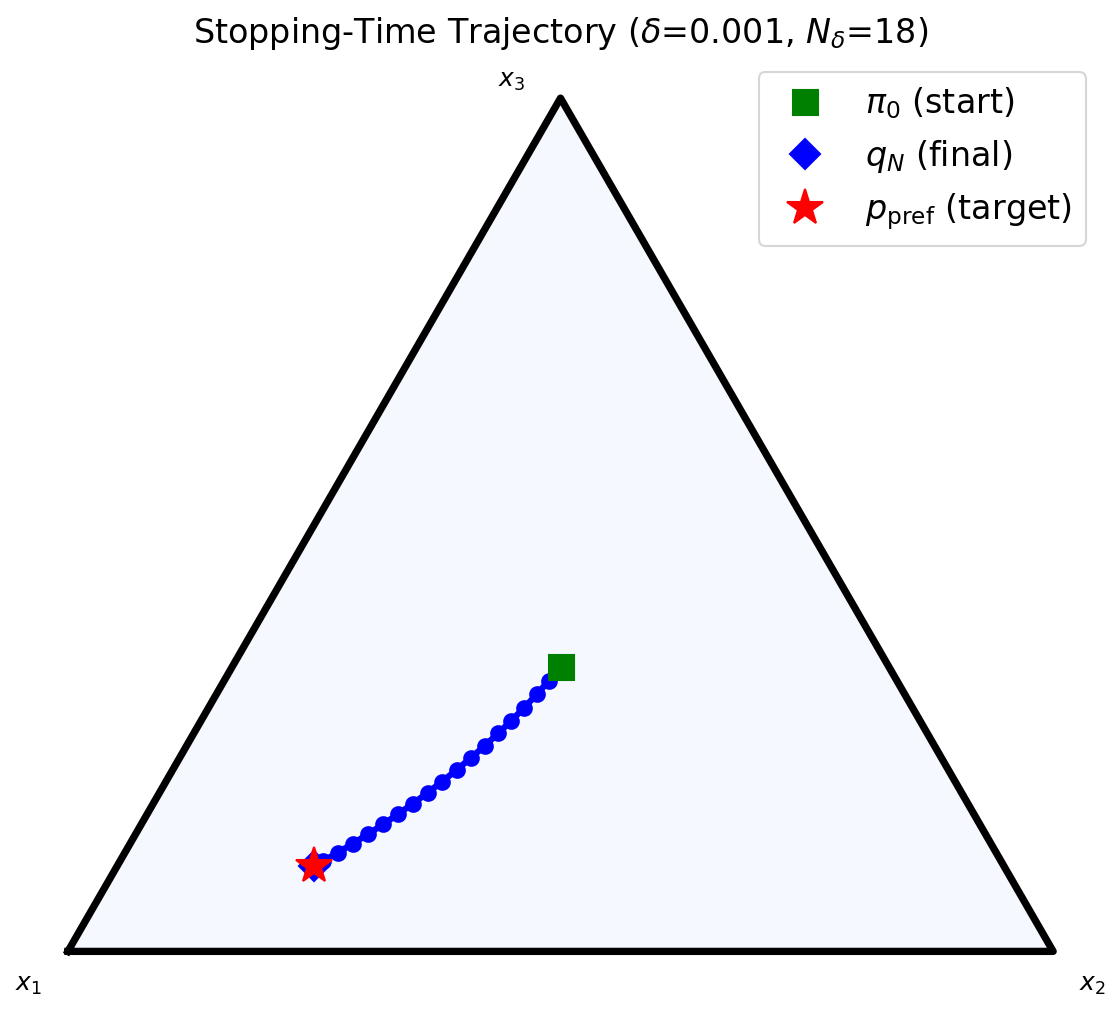}
\includegraphics[width=0.48\columnwidth]{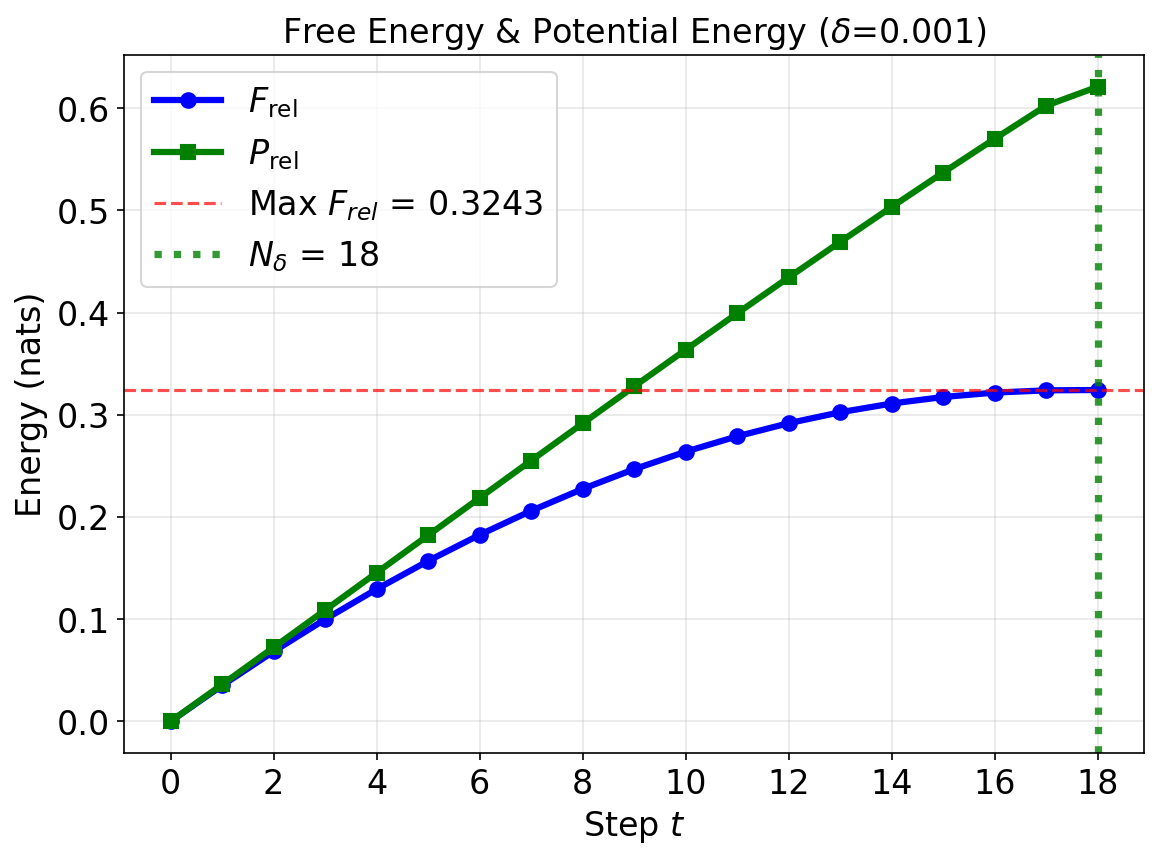}\\[4pt]
\includegraphics[width=0.48\columnwidth]{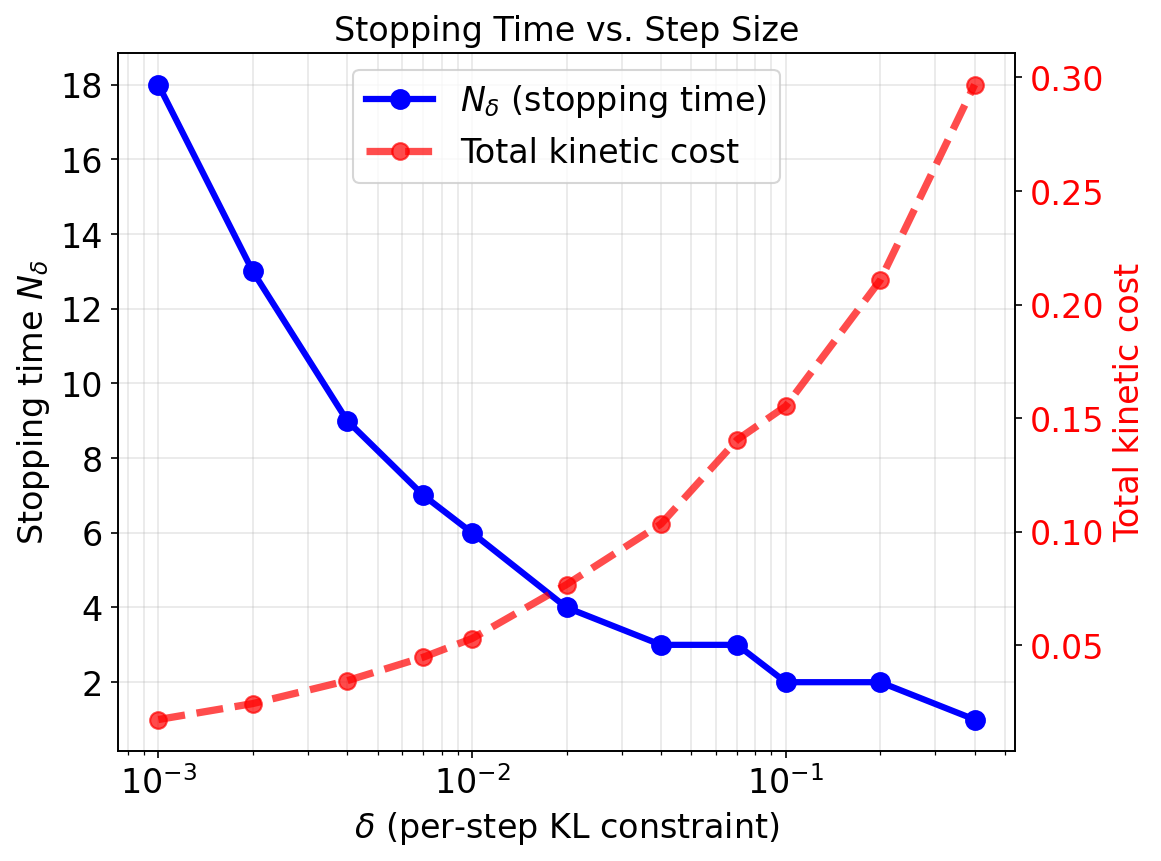}
\includegraphics[width=0.48\columnwidth]{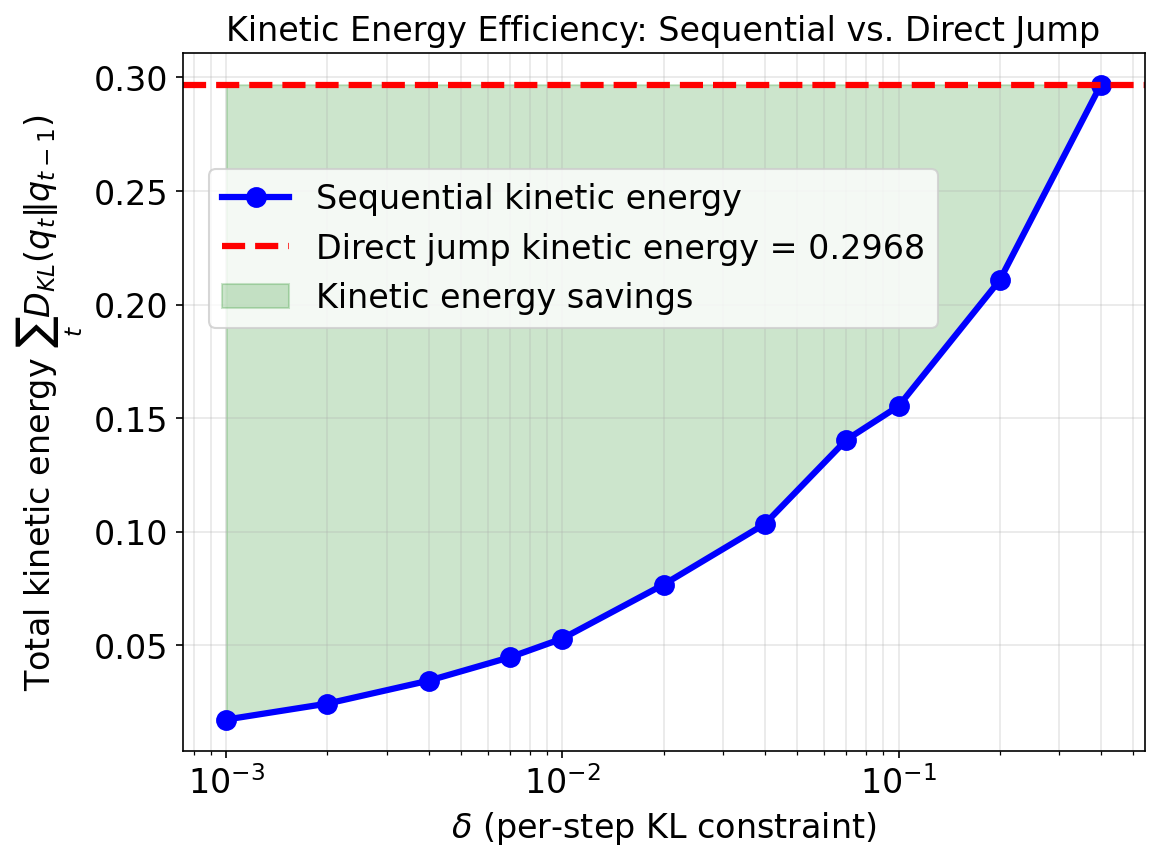}
\vspace{-0.1in}
\caption{Numerical results for $\pi_0 = (1/3, 1/3, 1/3)$, $p_{\mathrm{pref}} = (0.7, 0.2, 0.1)$. \textbf{Top-left:} Optimal trajectory on the $2$-simplex with $\delta = 0.001$ ($N_\delta = 18$ steps). \textbf{Top-right:} Evolution of relative free energy $F_{\mathrm{rel}}$ and relative potential energy $P_{\mathrm{rel}}$ along the path. \textbf{Bottom-left:} Stopping time $N_\delta$ and total kinetic cost as functions of the step-size constraint $\delta$. \textbf{Bottom-right:} Total kinetic energy of the sequential strategy versus the direct-jump cost $D_{\mathrm{KL}}(p_{\mathrm{pref}} \| \pi_0)$.}
\vspace{-0.2in}
\label{fig:numeric}
\end{figure}

\subsection{Numerical Illustration}
\label{sec:numeric}

We illustrate the stopping-time formulation with a concrete example on the $2$-simplex $\Delta(\mathcal{X})$ with $|\mathcal{X}|=3$. Let the initial equilibrium distribution be uniform, $\pi_0 = (1/3,\, 1/3,\, 1/3)$, and the preferred distribution be $p_{\mathrm{pref}} = (0.7,\, 0.2,\, 0.1)$. The initial KL divergences are $D_{\mathrm{KL}}(\pi_0 \| p_{\mathrm{pref}}) = 0.324$ nats and $D_{\mathrm{KL}}(p_{\mathrm{pref}} \| \pi_0) = 0.297$ nats. We sweep the kinetic constraint over $\delta \in [0.001, 0.4]$ and apply the stopping-time planner (Definition~4, Eq.~13). In all cases the trajectory terminates exactly at $p_{\mathrm{pref}}$.

\noindent \textbf{Trajectory and energy evolution:}
Figure~\ref{fig:numeric} (top-left) shows the least-action path on the simplex for $\delta = 0.001$. The trajectory traces a smooth curve from the centroid toward the target, consistent with the exponential tilting update of Theorem~2. Figure~\ref{fig:numeric} (top-right) displays the corresponding energy evolution: both $F_{\mathrm{rel}}$ and $P_{\mathrm{rel}}$ increase monotonically toward their maximum values, confirming steady progress toward the preferred distribution at each step.

\noindent \textbf{Effect of step size:}
Figure~\ref{fig:numeric} (bottom-left) reveals the tradeoff between step size and stopping time. As $\delta$ decreases from $0.4$ to $0.001$, the stopping time $N_\delta$ increases from $1$ to $18$, while the total kinetic cost decreases from $0.297$ to $0.017$ nats. This demonstrates that smaller steps are fundamentally more efficient, converting kinetic energy into free-energy reduction with less waste.

\noindent \textbf{Kinetic energy savings:}
Figure~\ref{fig:numeric} (bottom-right) compares the sequential kinetic cost against the direct-jump baseline $D_{\mathrm{KL}}(p_{\mathrm{pref}} \| \pi_0) = 0.297$. For all tested step sizes, the sequential strategy achieves strictly lower total kinetic energy. At $\delta = 0.001$, the savings exceed $94\%$, demonstrating the dominance of incremental motion established by the Pythagorean theorem (Theorem~\ref{thm:pyth}).

\section{Geodesic Cost on Variational Distribution}
\label{sec:geodesic}

\vspace{-0.1in}

The development so far assumes that motion on the probability simplex is
driven purely by the tradeoff between relative potential gain and
kinetic energy.  In many applications, however, there may be an
additional cost associated with occupying specific states or traversing
certain regions of the simplex.  Such costs may arise from energetic
penalties, environmental constraints, risk preferences, or application-
specific geodesic penalties.

In this section, we show that these additional costs can be naturally
incorporated into the framework without changing the essential
structure of the problem.  In particular, the modified optimization
reduces to the original one by replacing the preferred distribution
with an exponentially tilted version.

\subsection{Greedy Least-Action Problem with Geodesic Cost}

Let \(c(x)\ge 0\) denote a state-dependent cost associated with the
system occupying state \(x\in\mathcal X\).  At time \(t\), this induces
the expected state cost

\begin{equation}
\mathbb E_{q}[c(X)]
=
\sum_{x\in\mathcal X} q(x)c(x).
\label{eq:state_cost}
\end{equation}
While this term is included as an expected cost for occupying a state, it can equivalently be interpreted as a geodesic cost penalizing motion through specific regions of the probability simplex. Indeed $c(x)=c'(q(x))$, where $q(\cdot)$ is a distribution (point on the simplex) and $c'(\cdot)$ is the cost associated with that particular distribution. The greedy least-action problem therefore becomes
\begin{equation}
q_t^\star
=
\underset{q\in\Delta(\mathcal X)}{\text{argmin}} ~
\left\{
H(q\|p_{\mathrm{pref}})
+
\gamma \mathbb E_{q}[c(X)]
+
\lambda D_{\mathrm{KL}}(q\|q_{t-1})
\right\},
\label{eq:cost_problem}
\end{equation}
where \(\gamma\ge 0\) controls the importance of the geodesic/state
cost.

\subsection{Reduction to the Original Problem}

Expanding the first two terms,
\begin{comment}
\begin{align}
\label{eq:cost_expansion}
H(q_t\|p_{\mathrm{pref}}) + \gamma \mathbb E_{q_t}[c(X)] &=
%-\sum_x q_t(x)\log p_{\mathrm{pref}}(x) + \gamma \sum_x q_t(x)c(x)
%\nonumber\\
%&=
-\sum_x q_t(x) \left[ \log p_{\mathrm{pref}}(x) \gamma c(x) \right] \\
&=
-\sum_x q_t(x)
\log
\left[
p_{\mathrm{pref}}(x)e^{-\gamma c(x)}
\right].
\label{eq:cost_tilt}
\end{align}
\end{comment}
%Rewriting,
\begin{equation}
H(q\|p_{\mathrm{pref}})
+
\gamma \mathbb E_{q}[c(X)]
=
-\sum_x q(x)
\log
\left[
p_{\mathrm{pref}}(x)e^{-\gamma c(x)}
\right].
\label{eq:cost_tilt}
\end{equation}
\begin{comment}
Then, up to an additive constant independent of \(q_t\),
\begin{equation}
H(q_t\|p_{\mathrm{pref}})
+
\gamma \mathbb E_{q_t}[c(X)]
=
H(q_t\|\widetilde p_{\mathrm{pref}})
+\text{constant}.
\label{eq:equivalent_cross_entropy}
\end{equation}
%This immediately yields the following result.
%\begin{theorem}[Equivalence under geodesic cost]
Thus, 
\end{comment}
Problem \eqref{eq:cost_problem} is equivalent to
\begin{equation}
q_t^\star
=
\underset{q\in\Delta(\mathcal X)}{\text{argmin}} ~
\left\{ H(q\|\widetilde p_{\mathrm{pref}})
+
\lambda D_{\mathrm{KL}}(q\|q_{t-1})
\right\},
\label{eq:equivalent_problem}
\end{equation}
where
\(\widetilde p_{\mathrm{pref}}\)
is given by
%Define the exponentially tilted preferred distribution
\begin{equation}
\widetilde p_{\mathrm{pref}}(x)
\triangleq
\frac{
p_{\mathrm{pref}}(x)e^{-\gamma c(x)}
}{
\sum_{y\in\mathcal X}
p_{\mathrm{pref}}(y)e^{-\gamma c(y)}
}.
\label{eq:tilted_preferred_distribution}
\end{equation}
%\eqref{eq:tilted_preferred_distribution}.
%\end{theorem}
Hence, adding a geodesic cost does not change the basic structure
of the problem.  It simply modifies the preferred distribution through
an exponential tilt.

\subsection{Implications}

This result has a clean interpretation. The original framework attracts the system toward the preferred distribution \(p_{\mathrm{pref}}\).  Adding a state-dependent cost
modifies this attraction by suppressing expensive states through the
factor
\(
e^{-\gamma c(x)}.
\)
Consequently, states with large cost \(c(x)\) become less preferred and the geometry of the path changes through a deformation of the effective preferred distribution.

\begin{comment}
\begin{itemize}

\item states with large cost \(c(x)\) become less preferred,

\item states with low cost become relatively more attractive,

\item the geometry of the path changes through a deformation of the
effective preferred distribution.

\end{itemize}
\end{comment}
Geometrically, the state-dependent cost does not alter the structure of
the least-action path itself; rather, it shifts the effective attracting
point in the simplex from \(p_{\mathrm{pref}}\) to
\(\widetilde p_{\mathrm{pref}}\), causing the trajectory to bend toward
a cost-adjusted destination.
\begin{comment}
As a result, the tilted-distribution solution derived in the previous
section remains unchanged in form.  One simply replaces
\(p_{\mathrm{pref}}\) by
\(\widetilde p_{\mathrm{pref}}\), yielding
\begin{equation}
q_t^\star(x)
=
\frac{
q_{t-1}(x)\widetilde p_{\mathrm{pref}}(x)^{1/\lambda}
}{
\sum_y q_{t-1}(y)\widetilde p_{\mathrm{pref}}(y)^{1/\lambda}
}.
\label{eq:cost_tilted_solution}
\end{equation}
\end{comment}
Thus, the geodesic-cost extension preserves the least-action
structure of the original problem while simply modifying the
destination that attracts the path. In this sense, geodesic costs do not alter the physics of the
information-geometric least-action principle; they merely reshape the
effective potential landscape through exponential tilting of the
preferred distribution.  The structure mirrors entropy-regularized optimal transport~\cite{villani2009optimal,benamou2000computational}, where the optimal coupling similarly takes a Gibbs-type form.
\section{Discussion and Conclusion}

We have developed a variational framework for sequential motion on the
probability simplex, showing that belief evolution can be understood
through a discrete-time principle of least action.  Starting from an
initial Gibbs equilibrium distribution, the system evolves toward a
preferred distribution through a sequence of locally optimal updates
that balance two competing forces: attraction toward preferred states,
captured through relative potential energy, and resistance to abrupt
informational motion, captured through kinetic
energy.  The resulting least-action formulation reveals a thermodynamic
and geometric structure underlying belief evolution on the simplex.

A central result of the paper is that sufficiently small local moves
are fundamentally more efficient than large global jumps.  The
information-geometric Pythagorean theorem shows that the reduction in
free energy obtained through an optimal incremental move dominates the
corresponding informational step size, while the small-step symmetry of
KL divergence establishes that, in the infinitesimal regime, free-energy
reduction and kinetic expenditure become locally balanced.  This leads
naturally to a least-action principle in which optimal motion emerges
as a sequence of locally optimal variational steps.

The resulting update admits a closed-form solution: each new
distribution is obtained by exponentially tilting the current
distribution toward the preferred distribution.  Thus, the least-action
path is not a Euclidean interpolation on the simplex, but rather a
sequence of progressive information-geometric deformations that
continuously reshape the belief state until the preferred distribution
becomes reachable.  This provides a concrete geometric interpretation
of sequential belief updating and connects naturally to active
inference, where preferred outcomes guide belief evolution through
precision-weighted variational updates.

We further showed that state-dependent or geodesic costs can be
incorporated without altering the structure of the least-action
principle.  Such costs simply induce an exponential tilting of the
preferred distribution itself, revealing a structural equivalence
between external geometric penalties and a deformation of the effective
potential landscape.

Taken together, these results establish a physics of information
geometry in which motion on the probability simplex obeys a discrete
least-action principle analogous to classical mechanics, but with free
energy and informational kinetic cost replacing conventional potential
and kinetic terms.  The total action along a trajectory is bounded below by the free-energy difference between endpoints---the information-theoretic analogue of the Jarzynski equality~\cite{jarzynski1997nonequilibrium} and Crooks fluctuation theorem~\cite{crooks1999entropy}.
\begin{comment}
This perspective provides a new geometric and
thermodynamic foundation for belief dynamics, with potential
applications in active inference, variational learning, generative
modeling, and distributional control.

Future work includes extending the framework to partial observability
and imperfect measurements, establishing continuous-time limits through
Fisher--Rao geometry and gradient flows, and exploring empirical
applications in sequential inference and adaptive decision-making.
\end{comment}

\begin{comment}
\subsubsection{\discintname}
The authors have no competing interests to declare that are
relevant to the content of this article.
\end{credits}
\end{comment}

%
% ---- Bibliography ----
%
%\newpage
\bibliographystyle{splncs04}
\bibliography{bibliography}
%
% ---- Appendix ----
%
\begin{appendix}
\section{Proof of Lemma ~\ref{lemma:KL_symmetry}}
\label{sec:proof_lemma_KL_sym}

Using the Taylor expansion
\( 
\log(1+x)
=
x-\frac{x^2}{2}+O(x^3),
\)
we obtain
\begin{align*}
D_{\mathrm{KL}}(q_t^\star\|q_{t-1})
&=
\sum_x
(q_{t-1}(x)+\varepsilon_x)
\log\left(
1+\frac{\varepsilon_x}{q_{t-1}(x)}
\right)
%\\
=
\frac12
\sum_x
\frac{\varepsilon_x^2}{q_{t-1}(x)}
+
O(\|\varepsilon\|^3).
\end{align*}
Similarly,
\begin{align*}
D_{\mathrm{KL}}(q_{t-1}\|q_t^\star)
&=
\sum_x
q_{t-1}(x)
\log\left(
\frac{q_{t-1}(x)}
{q_{t-1}(x)+\varepsilon_x}
\right)
=
\frac12
\sum_x
\frac{\varepsilon_x^2}{q_{t-1}(x)}
+
O(\|\varepsilon\|^3).
\end{align*}
Subtracting the two expansions yields
\eqref{eq:kl_symmetry}. For the second part of the lemma, recall
\(
q_{\min}\triangleq \min_{x\in\mathcal X}q_{t-1}(x)>0.
\)
By Pinsker's inequality and the constraint
\(
D_{\mathrm{KL}}(q_t^\star\|q_{t-1})\le \delta,
\)
we have
\[
\|q_t^\star-q_{t-1}\|_1
\le
\sqrt{2D_{\mathrm{KL}}(q_t^\star\|q_{t-1})}
\le
\sqrt{2\delta}.
\]
If \(\delta\le q_{\min}^2/8\), then
\[
\|\epsilon\|_\infty
\le
\|\epsilon\|_1
\le
\sqrt{2\delta}
\le
\frac{q_{\min}}{2}.
\]
Hence, for every \(x\),
\[
\left|
\frac{\epsilon_x}{q_{t-1}(x)}
\right|
\le
\frac12.
\]
Let
\(
u_x=\frac{\epsilon_x}{q_{t-1}(x)}.
\)
Then
\[
q_t^\star(x)=q_{t-1}(x)(1+u_x),
\qquad
\sum_x q_{t-1}(x)u_x=0.
\]
Now expand the two KL divergences:
\[
D_{\mathrm{KL}}(q_t^\star\|q_{t-1})
=
\sum_x q_{t-1}(x)(1+u_x)\log(1+u_x),
\]
and
\[
D_{\mathrm{KL}}(q_{t-1}\|q_t^\star)
=
-\sum_x q_{t-1}(x)\log(1+u_x).
\]
Therefore,
\begin{equation}
D_{\mathrm{KL}}(q_t^\star\|q_{t-1})
-
D_{\mathrm{KL}}(q_{t-1}\|q_t^\star)
=
\sum_x q_{t-1}(x)(2+u_x)\log(1+u_x).
\end{equation}
Using
\[
\sum_x q_{t-1}(x)u_x=0,
\]
we may subtract the linear term and write
\[
D_{\mathrm{KL}}(q_t^\star\|q_{t-1})
-
D_{\mathrm{KL}}(q_{t-1}\|q_t^\star)
=
\sum_x q_{t-1}(x)\psi(u_x),
\]
where
\(
\psi(u)=(2+u)\log(1+u)-2u.
\)
For \(|u|\le 1/2\), Taylor's theorem gives
\(
|\psi(u)|\le C|u|^3
\)
for some constant \(C\). Hence
\[
\left|
D_{\mathrm{KL}}(q_t^\star\|q_{t-1})
-
D_{\mathrm{KL}}(q_{t-1}\|q_t^\star)
\right|
\le
C\sum_x q_{t-1}(x)|u_x|^3.
\]
Moreover, for \(|u_x|\le 1/2\),
\[
(1+u_x)\log(1+u_x)
\ge
\frac13 u_x^2,
\]
and therefore
\[
D_{\mathrm{KL}}(q_t^\star\|q_{t-1})
\ge
\frac13\sum_x q_{t-1}(x)u_x^2.
\]
Since
\[
D_{\mathrm{KL}}(q_t^\star\|q_{t-1})\le \delta,
\]
we obtain
\[
\sum_x q_{t-1}(x)u_x^2
\le
3\delta.
\]
Finally,
\[
\sum_x q_{t-1}(x)|u_x|^3
\le
\left(
\sum_x q_{t-1}(x)u_x^2
\right)^{3/2}
\le
(3\delta)^{3/2}.
\]
Thus,
\[
\left|
D_{\mathrm{KL}}(q_t^\star\|q_{t-1})
-
D_{\mathrm{KL}}(q_{t-1}\|q_t^\star)
\right|
\le
C(3\delta)^{3/2}.
\]
and
\[
D_{\mathrm{KL}}(q_{t-1}\|q_t^\star)
=
D_{\mathrm{KL}}(q_t^\star\|q_{t-1})
+
O(\delta^{3/2}),
\]
as \(\delta\to 0\).

\section{Proof of Theorem~\ref{thm:tilted}}
\label{sec:proof_thm_tilted}

From~\eqref{eq:action_cross_entropy}
%Starting from \eqref{eq:action_expanded}, combine the two terms:
\begin{align}
\mathcal{A}_t(q_t;q_{t-1})
&=
-\sum_x q_t(x)\log p_{\mathrm{pref}}(x)
+
\lambda_t
\sum_x q_t(x)\log\frac{q_t(x)}{q_{t-1}(x)}
\label{eq:action_expanded} \\
&=
\lambda_t
\sum_x q_t(x)
\log\frac{q_t(x)}{q_{t-1}(x)}
-
\sum_x q_t(x)\log p_{\mathrm{pref}}(x)
\nonumber\\
&=
\lambda_t
\sum_x q_t(x)
\left[
\log q_t(x)
-
\log q_{t-1}(x)
-
\frac{1}{\lambda_t}\log p_{\mathrm{pref}}(x)
\right]
\nonumber\\
&=
\lambda_t
\sum_x q_t(x)
\log
\frac{
q_t(x)
}{
q_{t-1}(x)p_{\mathrm{pref}}(x)^{1/\lambda_t}
}.
\label{eq:combined_terms}
\end{align}
Also defining the unnormalized tilted measure
%\begin{equation}
\( \widetilde q_t(x)
\triangleq
q_{t-1}(x)p_{\mathrm{pref}}(x)^{1/\lambda_t}, \)
%\label{eq:unnormalized_tilt}
%\end{equation}
and its normalizing constant
%\begin{equation}
\( Z_t(\lambda_t)
\triangleq
\sum_{y\in\mathcal X}
q_{t-1}(y)p_{\mathrm{pref}}(y)^{1/\lambda_t}, \)
%\label{eq:normalization_constant}
%\end{equation}
we have the normalized tilted distribution:
\begin{equation}
\bar q_t(x)
=
\frac{\widetilde q_t(x)}{Z_t(\lambda_t)}
=
\frac{
q_{t-1}(x)p_{\mathrm{pref}}(x)^{1/\lambda_t}
}{
\sum_{y\in\mathcal X}
q_{t-1}(y)p_{\mathrm{pref}}(y)^{1/\lambda_t}
}.
\label{eq:normalized_tilt}
\end{equation}

Using \(\widetilde q_t(x)=Z_t(\lambda_t)\bar q_t(x)\), we rewrite
\eqref{eq:combined_terms} as
\begin{align}
\mathcal A_t(q_t;q_{t-1}) &=
\lambda_t \sum_x q_t(x) \log \frac{ q_t(x)}{Z_t(\lambda_t)\bar q_t(x)}
\nonumber\\
&=
\lambda_t \sum_x q_t(x) \log \frac{q_t(x)}{\bar q_t(x)} - \lambda_t \log Z_t(\lambda_t) \sum_x q_t(x)
\nonumber\\
&=
\lambda_t
D_{\mathrm{KL}}(q_t\|\bar q_t)
-
\lambda_t \log Z_t(\lambda_t).
\label{eq:single_kl_form}
\end{align}

The second term in \eqref{eq:single_kl_form} is independent of \(q_t\).  Therefore minimizing \(A_t(q_t;q_{t-1})\) over the simplex is equivalent to minimizing
\( D_{\mathrm{KL}}(q_t\|\bar q_t). \)
Since KL divergence is nonnegative and satisfies
\[ D_{\mathrm{KL}}(q_t\|\bar q_t)=0 \quad
\Longleftrightarrow \quad q_t=\bar q_t, \]
the unique minimizer is \(q_t^\star=\bar q_t\), proving~\eqref{eq:tilted_distribution}.
\end{appendix}
\end{document}